\documentclass[12pt]{article}
\usepackage{a4}
\usepackage{latexsym}
\usepackage{amsmath}
\usepackage{amssymb}
\usepackage{tikz}
\usetikzlibrary{arrows.meta}
\usepackage{graphicx}
\usepackage{hyperref}

\usepackage{epic}

\newcommand{\bool}{{\bf B}}
\newcommand{\simcirc}{{\sf SIMCIRC}\ }

\newtheorem{theorem}{Theorem}

\newtheorem{lemma}[theorem]{Lemma}

\newenvironment{proof}{\noindent{\bf Proof:}}{$\Box$

\vspace{3 mm}}

\begin{document}
\title{Complexity of Simon's problem in classical sense}

\author{Hans Zantema\\
Eindhoven University of Technology, The Netherlands\\
email:  h.zantema@tue.nl \\
Radboud University Nijmegen, The Netherlands}

\maketitle

\begin{abstract}
Simon's problem is a standard example of a problem that is exponential in classical sense, while it admits a polynomial solution in quantum computing.
It is about a function $f$ for which it is given that a unique non-zero vector $s$ exists for which $f(x) = f(x \oplus s)$ for all $x$, where $\oplus$ is
the exclusive or operator. The goal is to find $s$.
The exponential lower bound for the classical sense assumes that $f$ only admits black box access. In this paper we investigate classical complexity when $f$
is given by a standard representation like a circuit. We focus on finding the vector space of all vectors $s$ for which $f(x) = f(x \oplus s)$ for all $x$,
for any given $f$. Two main results are: (1) if $f$ is given by any circuit, then checking whether this vector space contains a non-zero element is NP-hard,
and (2) if $f$ is given by any ordered BDD, then a basis of this vector space can be computed in polynomial time.
\end{abstract}

\section{Introduction}

Simon's problem \cite{S97} is the following. A function $f : \{0,1\}^n \to \{0,1\}^n$ is given with the particular property that exactly one non-zero $s \in \{0,1\}^n$ exists
such that $f(x \oplus s) = f(x)$ for all $x \in \{0,1\}^n$, where $\oplus$ stands for the {\em exclusive or} operation on each of the $n$ coordinates.
The goal is to find $s$. Although the problem
itself is purely artificial, it is of interest for the following reason. When $f$ is a {\em black box} function, that is, the only access to $f$ is by doing
queries: apply $f$ on a given input and observe the output, the worst case complexity to find $s$ in the classical sense is exponential in $n$. A basic observation
is that as soon as two distinct values $x,y$ are found with $f(x) = f(y)$ then $s = x \oplus y$, but finding such $x,y$ may take exponential time. On the other hand,
a {\em quantum circuit} can be designed that solves the problem in linear time \cite{S97}. In this way Simon's problem has become a standard example showing that quantum computing
may be exponentially faster than classical computing.

For this exponential gap it is crucial that $f$ is purely black box: nothing is known about how $f$ is computed. In contrast, in this paper we assume that $f$ is given
by some standard representation, and we investigate the corresponding classical complexity of Simon's problem. Before discussing the particular representations we give a
general observation on information content. For any non-zero $s \in \{0,1\}^n$ we can construct $2^{2^{n-1}}$ distinct functions $f$ having the required property as follows.
Let $i \in \{1,\ldots,n\}$ such that $s_i = 1$. Let $f_1 : \{x \in \{0,1\}^n \mid x_i = 0\} \to \{0,1\}^n$ be arbitrary, and define $f(x) = f_1(x)$ if $x_i = 0$ and
$f(x) = f_1(x \oplus s)$ if $x_i = 1$, note that $(x \oplus s)_i = 0$ if both $s_i = 1$ and $x_i = 1$. Now $f$ satisfies the required property. For every $s$ there are $2^{2^{n-1}}$
distinct functions $f_i$ all yielding distinct functions $f$ with the required property. Hence in any representation an arbitrary instance of such $f$ will require at least
$\log(2^{2^{n-1}})$ bits on average, being of exponential size. It makes sense to consider cases for which an algorithm to compute $f$ is given, and is of less
than exponential size.

A standard way to implement any function $f : \{0,1\}^n \to \{0,1\}^m$ is by a {\em circuit}: an acyclic network starting by $n$ input nodes, and in which values of
internal nodes are computed by the standard boolean operators $\neg$, $\vee$, $\wedge$. Eventually, $m$ resulting nodes serve as the output nodes, where $m=n$ in case of Simon's
problem. A natural question now is:
when $f$ is given by such a circuit, how hard is it to solve Simon's problem? A standard measure for hardness is being NP-hard. However, that is about decision problems, that is,
problems with a binary answer as output, and here the question is to find a vector $s$. A natural way to construct $s$ by a series of decision queries is by first
putting $s_1$ to 0, and then ask whether in the remaining problem on $s_2,\ldots,s_n$ admits an answer on Simon's problem. If it does, then $s_1$ is defined to be 0, otherwise
$s_1$ is defined to be 1: from the fact that $s$ exists and no $s$ satisfies $s_1 = 0$ indeed $s_1 = 1$ may be concluded.
This is repeated for $s_2,s_3,\ldots$ until the full vector $s$ has been determined. So the key decision problem in this approach is: given
any boolean function $f$, does there exist $s$ such that $f(x \oplus s) = f(x)$ for all $x$? One main result of this paper is that this decision problem is NP-hard.
Note that NP-hardness of this natural building block does not prove NP-hardness of Simon's problem itself, since this approach does not exploit the given property that
a corresponding non-zero $s$ uniquely exists.

A special kind of circuit is a {\em reduced ordered binary decision diagram} (ROBDD). A key property is that every boolean function has a unique representation as an ROBDD,
and that's one of the reasons that ROBDDs provide a standard data structure in symbolic model checking, VLSI design, and many other applications. A second main result of this paper
is that if $f$ is given by an ROBDD, then Simon's problem can be solved in polynomial time. In fact reducedness is not used here, so the results holds for any ordered BDD.

The paper is organized as follows.
First in Section \ref{secprel} we give some preliminaries: we formulate Simon's problem in terms of vector spaces over the Boolean domain. Next in Section \ref{seccirc} we present our NP-hardness
result for $f$ being presented by a circuit. Finally, in Section \ref{secbdd} we present our polynomial algorithm for $f$ being presented as an ordered BDD.

\section{Preliminaries}
\label{secprel}

We write $\bool = \{0,1\}$ for the Booleans, where $0,1$ are identified with false, true. The {\em exclusive or} operation (usually denoted by $\oplus$) we denote by $+$,
as it coincides with additional modulo 2. So $0+0 = 1+1 = 0$, $0+1 = 1+0 = 1$. For any $x,y \in \bool^n$ we write
$x+y = (x_1 + y_1, x_2 + y_2,\ldots,x_n+y_n)$ for $x = (x_1,\ldots,x_n)$ and $y = (y_1,\ldots,y_n)$.

For any function $f : \bool^n \to \bool^m$ we define
\[ V(f) = \{ s \in \bool^n \mid \forall x \in \bool^n : f(x + s) = f(x)\}. \]

With $+$ as addition and $\wedge$ as multiplication, $\bool$ is a field. It is straightforward to check that for  any function $f : \bool^n \to \bool^m$ the set $V(f)$ is a {\em vector space}
over $\bool$, that is, $\overline{0} = (0,0,\ldots,0) \in V(f)$, and $a+b \in V(f)$ for all $a,b \in V(f)$. We recall some basic notions and properties of such a vector space \cite{L97}:
\begin{itemize}
\item The size of any vector space $V \subseteq \bool^n$ is $2^k$ for some $k$ satisfying $0 \leq k \leq n$; this number $k$ is called the {\em dimension} of $V$.
\item A set $A \subseteq \bool^n$ is called {\em linearly independent} if no non-empty subset of $A$ has sum $\overline{0}$.
\item A set $A \subseteq V$ is called a {\em basis} of the vector space $V$ if it is linearly independent and every element of $V$ can be written as a sum of elements of $A$. Due to linear independence this representation is unique, in particular $\overline{0}$ is uniquely written as the empty sum.
\item For a vector space of dimension $k$ every basis has exactly $k$ elements.
\item A vector space $V$ is of dimension 0 if and only if $V = \{ \overline{0} \}$. A vector space $V$ is of dimension 1 if and only if $V = \{\overline{0}, s\}$ for some $s \in \bool^n \setminus \{ \overline{0} \}$.
\end{itemize}

So Simon's problem can be formulated as follows: given $f : \bool^n \to \bool^n$ for which $V(f)$ has dimension 1, find $s \neq \overline{0}$ such that $V(f) = \{\overline{0}, s\}$.

Such a function $f : \bool^n \to \bool^m$ can be seen as an $m$-tuple Boolean functions $f_1,\ldots,f_m : \bool^n \to \bool$, where $f(x) = (f_1,\ldots,f_m)$ for every $x \in \bool^n$.
From the definition of $V(f)$ it is clear that $V(f) = \bigcap_{i=1}^m V(f_i)$. So Simon's problem can be reformulated as follows: given $f : \bool^n \to \bool^n$ for which $V(f)$ has dimension 1, find $s \neq \overline{0}$ such that $s \in \bigcap_{i=1}^m V(f_i)$.

\section{Representing functions by circuits}
\label{seccirc}

A natural way to represent a function $f : \bool^n \to \bool^m$ is by a {\em circuit} with $n$ input nodes and $m$ output nodes, and several internal nodes of which the (boolean) values are
computed by logical gates in an acyclic way. Conceptually this is the same as stating that every output is a logical formula composed from the operators $\neg$, $\vee$ and $\wedge$ and
$n$ variables representing the input nodes, while in computing the size of such a logical formula we allow equal subformulas to be shared. So we represent circuits as being formulas,
one formula for every output node. For instance, the following circuit with three input nodes $x_1,x_2,x_3$ and two output nodes $f_1,f_2$ represents the function $f : \bool^3 \to \bool^2$
defined by
\[ f(x_1,x_2,x_3) = ((x_1 \wedge x_2) \vee \neg((x_1 \wedge x_2) \wedge x_3),\neg((x_1 \wedge x_2) \wedge x_3) \wedge \neg( x_2 \vee x_3)).\]

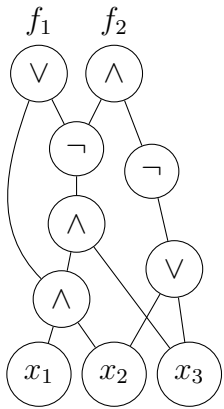
\begin{center}
\begin{tikzpicture}
\node[circle,draw] (n11) at (0.3,1) {$\wedge$};
\node[circle,draw] (n12) at (1.8,1.4) {$\vee$};
\node[circle,draw] (n21) at (0.5,2) {$\wedge$};
\node[circle,draw] (n22) at (1.5,2.7) {$\neg$};
\node[circle,draw] (n3) at (0.5,3) {$\neg$};
\node (f1n) at (0,4.7) {$f_1$};
\node (f2n) at (1,4.7) {$f_2$};
\node[circle,draw] (f1) at (0,4) {$\vee$};
\node[circle,draw] (f2) at (1,4) {$\wedge$};
\node[circle,draw] (x1) at (0,0) {$x_1$};
\node[circle,draw] (x2) at (1,0) {$x_2$};
\node[circle,draw] (x3) at (2,0) {$x_3$};
\draw (n11) -- (x1);
\draw (n11) -- (x2);
\draw (n12) -- (x2);
\draw (n12) -- (x3);
\draw (n22) -- (n12);
\draw (n21) -- (x3);
\draw (n21) -- (n11);
\draw (n3) -- (n21);
\draw (f1) -- (n3);
\draw (f1) edge[out=250,in=135,looseness=1] (n11);
\draw (f2) -- (n3);
\draw (f2) -- (n22);
\end{tikzpicture}
\end{center}

The {\em exclusive or} operator $+$ may also be used in formulas, where $p + q$ may be seen as an abbreviation of $(p \wedge \neg q) \vee (\neg p \wedge q)$, in which
in the circuit the instances of $p$ are shared, and similarly for $q$.

A circuit with $n$ input nodes and one output node is called {\em satisfiable} if it is possible to give Boolean values to the input nodes such that the output node yields true.
It is well-known \cite{HMU07} that checking satisfiability of such a circuit is NP-complete, that is, the problem is in NP and it is NP-hard.

The goal of this section is to investigate the hardness of Simon's problem in case the function $f : \bool^n \to \bool^n$ is given by a circuit with $n$ input nodes and $n$ output nodes.
Taking each of the $n$ output nodes separately, this can be seen as $n$ circuits for the functions $f_1,\ldots,f_n$, each having one output node,
and the goal is to find $s \in \bigcap_{i=1}^m V(f_i)$. One problem about stating NP-hardness is that Simon's problem
is about finding this Boolean vector $s$, while the notion of NP-hardness is about decision problems: problems with only a yes/no answer.
A natural approach to find the vector $s = (s_1,\ldots,s_n)$ by a series of calls to a decision problem is as follows.
First replace $s_1$ by 0, and then check whether $s = (0,s_2,\ldots,s_n) \neq \overline{0}$ exists in $V(f)$. If so, we fix $s_1$ to 0, otherwise we know that by the assumption that
$s \neq \overline{0}$ exists such that $s \in V(f)$ we may fix $s_1$ to 1. Next we do the same for $s_2$, and continue until $s = (s_1,\ldots,s_n)$ has been determined completely.
In this approach we apply \simcirc consecutively for values $n-1, n-2, \ldots$, for \simcirc defined as follows:

\begin{quote}
(\simcirc) Given a circuit  with $n$ input nodes and one output node, does there exist $s \neq \overline{0}$ such that $s \in V(f)$,
for $f$ being the function $f : \bool^n \to \bool$ defined by the circuit?
\end{quote}

We will prove that \simcirc is NP-hard. The key idea is to define a circuit $F(C)$ for any arbitrary single output circuit $C$ such that $C$ is unsatisfiable if and only if
this \simcirc question yields true for the circuit $F(C)$. Let $n$ be the number of inputs of $C$, let $p_1,\ldots,p_n$ be the $n$ variables corresponding to the $n$ input nodes of $C$.
In order to define the transformation $F$ we introduce $n+1$ fresh variables $p,q_1,\ldots,q_n$.
We define $h(C)$ to be a copy of $C$ in which every variable $p_i$ has been replaced by $p_i \wedge q_i$, and we define $D = (\bigwedge_{i=1}^n p_i) \wedge (\bigwedge_{i=1}^n \neg q_i)$.
Now $F(C)$ is defined by
\[ F(C) \; = \; (p \vee (h(C) + D)) \wedge (\neg p \vee D).\]
A main observation is that if $C$ is unsatisfiable, then $F(C)$ does not depend on $p$, yielding some $s \neq \overline{0}$ such that $s \in V(f)$,
being one direction of the desired key property of $F(C)$.
The ingredients $h(C)$ and $D$ of the definition of $F$ are chosen in order to be able to prove the reverse direction.

Now NP-hardness of \simcirc is an immediate consequence of the following theorem, combined with the observation that $F$ is polynomial.

\begin{theorem}
\label{thmcirc}
  Let $C$ be an arbitrary circuit with $n$ inputs and one output. Then $C$ is unsatisfiable if and only if there exists $s \neq \overline{0}$ such that $s \in V(f)$, for $f$ being the function $f : \bool^{2n+1} \to \bool$
  defined by the circuit $F(C)$.
\end{theorem}

\begin{proof}
First assume that $C$ is unsatisfiable. Then both $C$ and $h(C)$ are equivalent to false, and $h(C) + D$ is equivalent to $D$, and also $F(C)$ is equivalent to $D$.
Let $s \neq \overline{0}$ be the input vector in which $p$ is true and all other variables $p_i,q_i$ are false. Then $x + s$ is a copy of $x$ in which the value of $p$ is swapped,
while the values of all other variables are kept. As $F(C) \equiv D$ and $p$ does not occur in $D$, we obtain  $s \in V(f)$, proving the 'only if' part of the theorem.

Conversely, assume that  $s \in V(f)$ for some $s \neq \overline{0}$, that is, $s \in \bool^{2n+1}$ and $F(C)(x) = F(C)(x+s)$ for all $x \in \bool^{2n+1}$.
We have to prove that $C$ is unsatisfiable, that is, equivalent to false. We consider two cases depending on the value of $s_p$, where $s_p$ is the coordinate of
$s \in \bool^{2n+1}$ corresponding to the variable $p$.

(1) First assume that $s_p$ is false. Choose an arbitrary $x \in \bool^{2n+1}$ such that $x_p = 1$, hence also $x_p + s_p = 1$. Then $D(x) = F(C)(x) = F(C)(x+s) = D(x+s)$. Note that
$D = (\bigwedge_{i+1}^n p_i) \wedge (\bigwedge_{i+1}^n \neg q_i)$, and $D(x+s) = D'(x)$ for $D'$ obtained from $D$ by replacing at least one of the variables $p_i,q_i$ by its negation.
As this holds for every $x \in \bool^{2n+1}$ for which $x_p$ is true, we obtain that $D$ and $D'$ are equivalent, contradiction.

(2) Next assume the remaining case in which $s_p$ is true. Choose an arbitrary $x \in \bool^{2n+1}$ with $x_p$ is false. Then $h(C)(x) + D(x) = F(C)(x) = F(C)(x+s) = D(x+s)$. So
$h(C)(x) = D(x) + D(x+s)$. We consider $h(C)$ and $D$ as circuits over $p_1,\ldots,p_n,q_1,\ldots,q_n$, and have to prove that $h(C)$ is unsatisfiable. If it is not, then we can choose
$x$ such that $h(C)(x) = D(x) + D(x+s) = 1$. That is, $D(x) = 1$ and $D(x+s) = 0$, or conversely. In the latter case we replace $x$ by $x+s$. So without loss of generality
we may assume that $D(x) = 1$
and $D(x+s) = 0$. From $D(x) = 1$ and the definition of $D$ we conclude $x_{p_i} = 1$  and $x_{q_i} = 0$ , for all $i = 1,\ldots,n$. Now we modify $x$ to $x'$ and $x''$ as follows.
\[ x'_{p_1} = 0, \; x'_{q_1} = 0, \; x''_{p_1} = 0, \; x''_{q_1} = 1, \; x'_{p_i} = x''_{p_i} = 1, x'_{q_i} = x''_{q_i} = 0 \mbox{ for $i>1$}.\]
Since $h(C)$ only depends on $p_i \wedge q_i$, and the value of $p_i \wedge q_i$ does not change by replacing $x$ by $x', x''$, both $h(C)(x')$ and $h(C)(x'')$ yield true just like $h(C)(x)$.
Since $h(C)(x') = D(x') + D(x'+s) = 1$ and $D(x') = 0$, we obtain $D(x'+s) = 1$. Similarly we obtain $D(x''+s) = 1$. But now we have two distinct boolean vectors $x'+s$ and $x''+s$ for which
$D$ both yield true, contradicting the definition of $D$.

\end{proof}

In our version of \simcirc we considered functions of which the output is a single boolean. But this generalizes easily: the same result is obtained for functions
$f : \bool^n \to \bool^m$ for a fixed $m$, or functions $f : \bool^n \to \bool^n$ as they occur in Simon's problem,
simply by taking multiple output circuits in which all outputs are equal.

Although \simcirc seems to be a natural building block for solving Simon' problem, Theorem \ref{thmcirc} does not prove that Simon's problem itself is NP-hard.
A main observation is that in using \simcirc
to solve Simon's problem, the given fact that $V(f)$ contains a unique $s \neq \overline{0}$ is not exploited. However, by using this unicity, the problem is in the class UP \cite{V76}
(unambiguous polynomial) rather than in NP, and UP-complete problems are not known.

\section{BDDs}
\label{secbdd}

A standard data structure for representing boolean functions is given by {\em BDD}s: boolean decision diagrams. Introduced in \cite{B86}, they have several applications in,
e.g., symbolic model checking and VLSI design (\cite{MT98,CHVB18}). Surprisingly, they can be presented as a special instance of circuits by adding additional operators.
For every variable $p$ we introduce an additional binary operator, also denoted by $p$ in a prefix notation to be used in circuits / formulas,
in which the meaning of $p(A,B)$ is defined to be $(\neg p \vee A) \wedge (p \vee B)$. Intuitively, the meaning of $p(A,B)$ is:
if $p$ then $A$ else $B$. It is straightforward to see that the notion of BDD as introduced in \cite{B86,MT98,CHVB18} coincides with a circuit only composed from false,
true, and the operators $p$ for $p$ running over all variables. Hence here we identify BDDs with such circuits.
A BDD is called {\em reduced} if it is
optimally shared, that is, no two distinct nodes represent the same boolean function. A BDD is called {\em ordered} with respect to a total order $<$ on the variables if for every node
of the form $p(A,B)$, the nodes $A,B$ are false, true, or of the form $q(-,-)$ with $p < q$. So for every node in an ordered BDD all nodes below are larger.
The notion of reduced ordered BDD is abbreviated to ROBDD. A key property of ROBDDs is {\em unicity}, see e..g, \cite{MT98}, Cor 6.12:
\begin{quote}
For a fixed order $<$ on variables, every boolean function on these variables has a unique representation as ROBDD with respect to $<$.
\end{quote}
For our result this uniqueness and being reduced does not play a role, but being ordered does. So we will focus on ordered BDDs.

\vspace{1mm}

\noindent \begin{minipage}{105mm}
An example of an ROBDD on the variables $p,q,r$ with respect to the order $p < q < r$ is given on the right; later on we will use it to illustrate our main algorithm.
The top node is the output.
To compute the value of the corresponding boolean function we start at this top node.
For every node we go to the left if the value of the corresponding variable is true, and to the right if this value is false, following the if-then-else meaning of the nodes.
This is continued until one of the leaves 0 or 1 is reached.
If it is 0, then the function result is false, if it is 1 then the function result is true.
\end{minipage}\hspace{5mm}
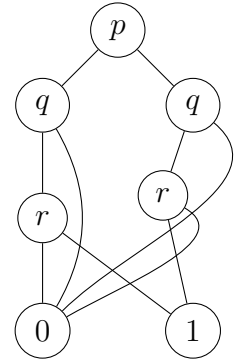
\begin{minipage}{4cm}
\begin{tikzpicture}
\node[circle,draw] (p) at (1,4) {$p$};
\node[circle,draw] (q1) at (0,3) {$q$};
\node[circle,draw] (q2) at (2,3) {$q$};
\node[circle,draw] (r1) at (0,1.5) {$r$};
\node[circle,draw] (r2) at (1.6,1.8) {$r$};
\node[circle,draw] (0) at (0,0) {$0$};
\node[circle,draw] (1) at (2,0) {$1$};
\draw (p) -- (q1);
\draw (p) -- (q2);
\draw (q1) -- (r1);
\draw (q2) -- (r2);
\draw (r1) -- (0);
\draw (r1) -- (1);
\draw (r2) -- (1);
\draw (r2) edge[out=330,in=30,looseness=1] (0);
\draw (q1) edge[out=300,in=60,looseness=1] (0);
\draw (q2) edge[out=320,in=45,looseness=1] (0);
\end{tikzpicture}
\end{minipage}

\vspace{1mm}

We may also consider BDDs with multiple output nodes,
representing functions $\bool^n \to \bool^m$, where $n$ is the number of variables and $m$ is the number of output nodes.
Just like in any circuit, every node represents a function $\bool^n \to \bool$.

Any BDD corresponds to a classical circuit representing the same function, obtained by expanding every occurrence of $p(A,B)$ by $(\neg p \vee A) \wedge (p \vee B)$.
In contrast to arbitrary circuits for which we showed that even the basic question whether $V(f)$ contains a non-zero vector is already NP-hard, for ordered BDDs this turns out to be
feasible in polynomial time. Not only this decision question is feasible, even computing a full basis, not only for ordered BDDs with a single output node, but also for any number of
output nodes. For doing so, we need some standard algorithms from linear algebra:
\begin{itemize}
\item When two vector spaces are given by a basis of each of them, a basis of the intersection of the vector spaces can be computed by the standard Zassenhaus algorithm in polynomial
time, see \cite{LRW97}.
\item Exploiting this algorithm, if $V,W$ are vector spaces and $v,w$ are vectors, then it can be checked whether $v+V \cap w+W$ is empty, and if not, an element in $v+V \cup w+W$ can be determined.
All of this in polynomial time.
\end{itemize}

For presenting the algorithm we need the notion $V(T,U)$ for $T,U$ being nodes in an ordered BDD, defined as follows:
\[ V(T,U) = \{ s \mid T(x) = U(x+s) \mbox{ for all $x$} \}.\]
It satisfies the following properties.
\begin{lemma}
\label{lembdd}
For every two nodes $T,U$ in an ordered BDD we have
\begin{enumerate}
\item $V(T,U) = V(U,T)$.
\item If $T,U$ are labeled by distinct variables, then $V(T,U) = \emptyset$.
\item If $s \in V(T,U)$ then $V(T) = V(U)$ and $V(T,U) = s + V(T)$.
\end{enumerate}
\end{lemma}
\begin{proof}
1 and 3 directly follow from the definitions. For 2 let $T$ and $U$ be labeled by distinct variables, and assume $s \in V(T,U)$. Then one among $T,U$ depends on the smaller of these
two variables while the other does not, contradicting $T(x) = U(x+s)$ for all $x$.
\end{proof}


\begin{theorem}
For an ordered BDD with respect to any order $<$, representing a boolean function $f : \bool^n \to \bool^m$, there is an algorithm that computes a basis for the vector space $V(f)$ in
time polynomial in the size of the BDD.
\end{theorem}

\begin{proof}
Let $p_1 < p_2 < \cdots < p_n$ be the $n$ variables. We identify each node $T$ of the BDD by the boolean function it represents. The algorithm runs over all nodes labeled
by $p_i$, for $i$ running down from $n$ to 1, and controlled by Lemma \ref{lembdd} it computes
\begin{enumerate}
\item whether $V(T,U)$ is empty for every two nodes $T,U$ labeled by $p_i$, and if not then it computes an element $s_{TU} \in V(T,U)$;
\item a basis of $V(T)$ for every node $T$ labeled by $p_i$.
\end{enumerate}
As $T$ labeled by $p_i$ only depends on $p_i,p_{i+1},\ldots,p_n$, $V(T)$ is considered as a subspace of the corresponding $n-i+1$-dimensional vector space.

As the computation is done for $i$ running down from $n$ to 1, during the process a node labeled by $p_i$ is considered when a basis
of $V(T)$ has been computed already for all nodes $T$ with label $p_j$, $j> i$, and for all nodes with label $> p_i$ the corresponding vectors $s_{TU}$ are known.

First for any two distinct nodes $T,U$ labeled by $p_i$ a vector $s_{TU} \in V(T,U)$ is computed if it exists. Let $T = p_i(T_1,T_2)$ and $U = p_i(U_1,U_2)$. For the new coordinate of $s_{TU}$
corresponding to the variable $p_i$ there are two options: either 0 or 1. First consider adding 0. Then both $V(T_1,U_1)$ and $V(T_2,U_2)$ should be non-empty, and $s_{T_1U_1} \in V(T_1,U_1)$
and $s_{T_2U_2} \in V(T_2,U_2)$ have been computed already. In case the root of $T_1$ is $p_{i+1}$ then $s_{T_1U_1}$ has the right dimension, if it is $p_j$ for $j > i+1$ then $j-i-1$ zeros
are added to the vector $s_{T_1U_1}$ before the new 0 for $p_i$ is added, and similarly for $T_2$.

Now the algorithm checks whether
$s_{T_1U_1} + V(T_1) \cap s_{T_2U_2} + V(T_2)$ is empty. If it is not, then $s_{TU}$ is defined to be an element of this intersection, extended by a 0 for the new position $p_i$.
One checks that indeed $s_{TU} \in V(T,U)$.
If this fails, then the other option is that the value for $p_i$ is 1. Then both $s_{T_1U_2}$ and $s_{T_2U_1}$ should exist, and extended by zeros if necessary. In that case it is checked whether
$s_{T_1U_2} + V(T_1) \cap s_{T_2U_1} + V(T_2)$ is empty. If it is not, then $s_{TU}$ is defined to be an element of this intersection, extended by a 1 for the new position $p_i$.
One checks that indeed $s_{TU} \in V(T,U)$. In all other cases $V(T,U)$ is empty.

Next a basis for $V(T)$ is computed for every $T$ labeled by $p_i$, that is, $T = p_i(T_1,T_2)$. In case the root of $T_1$ is $p_j$ with $j > i+1$ then $j-i-1$ zeros are added
to all vectors in the basis of $V(T_1)$, and similarly for $T_2$. Now a basis for $V(T_1) \cap V(T_2)$ is computed, and every vector in the basis is extended
by an extra 0 for the new position $p_i$. In case $s_{T_1T_2}$ does not exist then this is a basis for $V(T)$. In case $s_{T_1T_2}$ exists then one more vector is added to the basis, namely
$s_{T_1T_2}$ extended by an extra 1 for the new position $p_i$. In both cases a basis for $V(T)$ has been computed.

This is done for all $i$, running down from $n$ to 1. Then finally a basis of $T(f)$ is  obtained as a basis of the intersection of all $V(T)$, for $T$ running over the $m$ output nodes
of the BDD.
\end{proof}

\noindent\begin{minipage}{105mm}

To sketch what is going on in the algorithm we consider the same ROBDD with respect to $p  < q < r$ that we saw before.
First for the two nodes $T_r,U_r$ labeled by $r$ the one-dimensional vector $s_{T_rU_r}$ is computed, being equal to $s_{U_rT_r}$. The value 0 for the position $r$ fails since $s_{01}$ does not exist, but
the value 1 for the position $r$ succeeds since $s_{00} = s_{11} = \epsilon$. So $s_{T_rU_r} = (1)$: a vector of length 1 corresponding to the variable $r$, expressing that by swapping
$r$, the two $r$-nodes transform to each other. Next bases of $V(T_r)$ and $V(U_r)$ are computed, both being empty. Let $T_q,U_q$ be the two nodes labeled by $q$,
the next step is to compute $s_{T_qU_q} = s_{U_qT_q}$.
\end{minipage}\hspace{2mm}
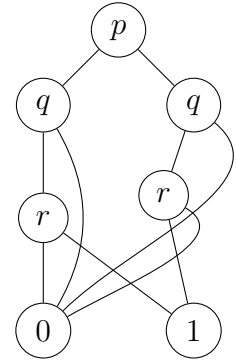
\begin{minipage}{4cm}

\begin{center}
\begin{tikzpicture}
\node[circle,draw] (p) at (1,4) {$p$};
\node[circle,draw] (q1) at (0,3) {$q$};
\node[circle,draw] (q2) at (2,3) {$q$};
\node[circle,draw] (r1) at (0,1.5) {$r$};
\node[circle,draw] (r2) at (1.6,1.8) {$r$};
\node[circle,draw] (0) at (0,0) {$0$};
\node[circle,draw] (1) at (2,0) {$1$};
\draw (p) -- (q1);
\draw (p) -- (q2);
\draw (q1) -- (r1);
\draw (q2) -- (r2);
\draw (r1) -- (0);
\draw (r1) -- (1);
\draw (r2) -- (1);
\draw (r2) edge[out=330,in=30,looseness=1] (0);
\draw (q1) edge[out=300,in=60,looseness=1] (0);
\draw (q2) edge[out=320,in=45,looseness=1] (0);
\end{tikzpicture}
\end{center}
\end{minipage}

\vspace{1mm}

The algorithm computes the intersection of $(1)+V(T_{q1})$ consisting of the single vector $(1)$ and $s_{00} + V(T_{q2})$,
which is the full one-dimensional space, so the intersection consists of $(1)$. Extended by a 0 in front for the new position for variable $q$ this yields $s_{T_qU_q} = s_{U_qT_q} = (0,1)$.
Both $V(T_q)$ and $V(U_q)$ only consist of the zero vector $(0,0)$, and have an empty basis.

Finally, a basis for $V(T_p)$ is computed, for $T_p$ being the root node labeled by $p$. It starts by taking a basis for $V(T_q \cap V(U_q)$, and then it is checked whether
$s_{T_qU_q}$ exists. Indeed it exists and equals $(0,1)$, to be extended by 1 in front. So the empty basis is extended by the vector $(1,0,1)$. So the resulting basis of $V(T_p)$ is $(1,0,1)$.
Indeed, the effect of addition of $(1,0,1)$ is that both $p$ and $r$ are swapped, by which the BDD is transformed to itself, and according to the theorem apart from the identity
this is the only operation doing so.


\bibliography{simonarx}

\end{document}